\documentclass[11pt]{article}

\usepackage{amsmath}
\usepackage{amsfonts}
\usepackage{amsthm}
\usepackage{amssymb}
\usepackage[utf8]{inputenc}
\usepackage[noadjust]{cite}

\usepackage{textcomp}

\usepackage{xcolor}

\usepackage[margin=1in]{geometry}
\usepackage{indentfirst}

\usepackage{authblk}
\usepackage{abstract}

\usepackage{pbox}
\usepackage{comment}

\usepackage[linesnumbered, boxruled, noend]{algorithm2e}
\usepackage{pgfplots}
\pgfplotsset{compat=1.4}
\usetikzlibrary{shapes}
\usetikzlibrary{plotmarks}

\usepackage{framed}
\usepackage{mdframed}
\usepackage{placeins}
\usepackage{subcaption}
\usepackage{afterpage}
\usepackage{hyperref}
\usepackage{cleveref}

\usepackage{booktabs}

\usepackage{courier}

\newcommand{\eps}{\varepsilon}
\newcommand{\Oish}{\widetilde{O}}

\newcommand{\rpp}{\texttt{RP}}
\newcommand{\zz}{\mathbb{Z}}
\newcommand{\rr}{\mathbb{R}}
\newcommand{\Chlamtac}{Chlamt\'a\v{c}}

\DeclareMathOperator{\dist}{dist}
\DeclareMathOperator{\rs}{\mbox{\tt RS}}

\newtheorem{theorem}{Theorem}

\newtheorem{lemma}{Lemma}
\newtheorem{definition}{Definition}

\newtheorem{claim}{Claim}

\newtheorem*{definition*}{Definition}
\newtheorem*{theorem*}{Theorem}

\newtheorem*{lemma*}{Lemma}

\title{Reachability Preservers: New Extremal Bounds and Approximation Algorithms\footnote{Short Version Appeared in SODA '18.  Work performed while authors were employed by Stanford University.}}
\author[1]{Amir Abboud\thanks{amir.abboud@weizmann.ac.il} }
\author[2]{Greg Bodwin\thanks{bodwin@umich.edu} }
\affil[1]{Weizmann Institute}
\affil[2]{University of Michigan}
\date{}

\begin{document}

\maketitle
\thispagestyle{empty}

\begin{abstract}
We define and study \emph{reachability preservers}, a graph-theoretic primitive that has been implicit in prior work on network design.
Given a directed graph $G = (V, E)$ and a set of \emph{demand pairs} $P \subseteq V \times V$, a reachability preserver is a sparse subgraph $H$ that preserves reachability between all demand pairs.

Our first contribution is a series of extremal bounds on the size of reachability preservers.
Our main result states that, for an $n$-node graph and demand pairs of the form $P \subseteq S \times V$ for a small node subset $S$, there is always a reachability preserver on $O(n+\sqrt{n |P| |S|})$ edges.
We additionally give a lower bound construction demonstrating that this upper bound characterizes the settings in which $O(n)$ size reachability preservers are generally possible, in a large range of parameters.

The second contribution of this paper is a new connection between extremal graph sparsification results and classical Steiner Network Design problems.
Surprisingly, prior to this work, the osmosis of techniques between these two fields had been superficial.
This allows us to improve the state of the art approximation algorithms for the most basic Steiner-type problem in directed graphs from the $O(n^{0.6+\eps})$ of \Chlamtac{}, Dinitz, Kortsarz, and Laekhanukit (SODA'17) to $O(n^{4/7+\eps})$.
\end{abstract}
\setcounter{page}{0}

\pagebreak

\section{Introduction}
\label{sec:intro}

In this paper we prove new results about the extremal structure of paths in directed graphs.
As a motivating example, suppose we are given an $n$-node directed graph $G$, a set of source nodes $S$ of size $|S|=n^{1/3}$, and a subset $P \subseteq S \times V$ of ``demand pairs'' of size $|P| = O(n^{2/3})$, such that for all demand pairs $(s, t)$ there exists an $s \leadsto t$ path in $G$.
Our goal is to remove as many edges from $G$ as possible while maintaining reachability for all demand pairs.
How many edges will we have to keep?
It is not hard to see that $O(n^{4/3})$ edges will suffice: for each source $s \in S$ we can keep a BFS tree at the cost of $\le n$ edges, and this will guarantee that $s$ still reaches all the nodes it used to reach.
But can we improve this $O(n^{4/3})$ bound to $O(n)$ in general?
Or are there instances where we are forced to keep at least $\Omega(n^{4/3})$ edges?
Or is the right answer somewhere in between?

Graph reachability is almost as basic of a notion as directed graphs themselves. 
It is ubiquitous in math, science, and technology.
Computational questions related to graph reachability are central to various fields. 
For example, the classical $\mathsf{NL}$ vs. $\mathsf{L}$ open question asks if one can find a directed path using small space.
We would arguably be in a much better shape for tackling all the fundamental questions involving reachability if we could give good answers to basic structural questions like the one above. 

To study questions along these lines, we abstract and study \emph{reachability preservers}:
\begin{definition} [Reachability Preservers]
Given a graph $G = (V, E)$ and a set of demand pairs $P \subseteq V \times V$, a subgraph $H$ is a \emph{reachability preserver} of $G, P$ if for all demand pairs $(s, t) \in P$, if there exists an $s \leadsto t$ path in $G$, then there also exists an $s \leadsto t$ path in $H$.
\end{definition}

In the following, we outline some of our results on the possibilities and limitations of constructing sparse reachability preservers of arbitrary input graphs.

\subsection{Reachability Preserver Upper Bounds}

Our main positive result is the following theorem, which improves on all previously known upper bounds by polynomial factors. 

\begin{theorem}[Sparse Reachability Preservers]
\label{thm:main}
For any $n$-node graph $G=(V,E)$, set of source nodes $S \subseteq V$, and set of demand pairs $P \subseteq S \times V$, there is a reachability preserver $H$ of size
$$|E(H)| = O\left(n+\sqrt{n |P| |S|}\right).$$
\end{theorem}

For comparison to prior work, it was known that all distances can be preserved (not just reachability) using
$$|E(H)| = O\left(\min\{ n^{2/3} |P| , n |P|^{1/2}, n|S| \}\right)$$
edges \cite{Bodwin21sicomp,CE06}.
In our original motivating example, these distance preserver bounds only give $|E(H)| = O(n^{4/3})$.
Yet in this same example, our new bound implies $|E(H)| = O(n)$.
In general, our theorem states that whenever $|P| = o(n |S|)$, there are much better reachability preservers than obtained by simply keeping spanning trees out of every source node.
Our main result also implies a new upper bound for reachability preservers in the general setting, where the demand pairs aren't necessarily source-restricted:
\begin{theorem} \label{thm:intropairwise}
For any $n$-node graph and set of demand pairs $P$, there is a reachability preserver $H$ of size
$$|E(H)| = O\left( n + (n|P|)^{2/3} \right).$$
\end{theorem}

%

The previous best bound was
$$|E(H)| = O\left( \min\left\{np^{1/2}, n^{1/2}p\right\} + n\right)$$
from the work of Coppersmith and Elkin \cite{CE06}, although this bound applies to subgraphs $H$ that more strongly try to preserve distances.\footnote{In more detail: one can straightforwardly reduce reachability preservers to the case where the input graph is a DAG, and although not explicitly stated in \cite{CE06}, their argument implies this upper bound for distance preservers in DAGs.}
As stated above, Theorem \ref{thm:main} is a purely existential bound, so we next discuss the efficient construction of sparse reachability preservers.
The folklore \emph{decremental greedy algorithm} suffices to build reachability preservers of optimal size in polynomial time: that is, starting with the input graph $G$, continually find and remove an edge $e$ as long as $G \setminus e$ still preserves reachability among all demand pairs.
Once this process halts, the final subgraph is the unique reachability preserver of itself, so by Theorem \ref{thm:main} it must have only $O(n + \sqrt{n|P||S|})$ edges.
Our next result demonstrates that there is actually a far more efficient polynomial-time algorithm than decremental-greedy:

\begin{theorem} [Fast Construction of Reachability Preservers] \label{thm:reach-construct-intro}
For an input graph $G = (V, E)$, set of source nodes $S$, and set of demand pairs $P$, there is a randomized algorithm that runs in $\Oish(|E||S|)$ time and computes a reachability preserver as in Theorem \ref{thm:main} with high probability.
\end{theorem}

A similar fast algorithm holds for Theorem \ref{thm:intropairwise}, although the runtime will still depend on the number of source nodes used by the demand pairs, even though the edge bound in Theorem \ref{thm:intropairwise} does not.
This algorithm (like the decremental greedy algorithm) actually has a stronger property called \emph{existential optimality}: essentially, if one can improve the upper bound for reachability preservers in Theorems \ref{thm:main} or \ref{thm:intropairwise}, then our algorithm automatically produces reachability preservers of this new improved size.

\subsection{Approximation Algorithms: The Directed Steiner Network Problem}

Let us consider an input to Theorem \ref{thm:main} where $|S|=n^{3/4}$ and $|P|=n^{5/4}$.
Theorem \ref{thm:main} guarantees the existence of a preserver on $O(n^{1.5})$ edges, and Theorem \ref{thm:reach-construct-intro} says that we can find this preserver efficiently.
It is easy to observe that, in \emph{the worst case} over all graphs $\Omega(|P|)=\Omega(n^{5/4})$ edges could be necessary, e.g. if the graph is a directed biclique.
However, from a real-world point of view, why should we expect our graphs to be worst case?
It could be that a particular input instance $G, S, P$ enjoys a much sparser reachability preserver than what extremal results like Theorem \ref{thm:main} can guarantee.

So, let us denote the number of edges in the sparsest possible reachability preserver of our given input instance by $OPT$.
Are there efficient algorithms that can find a reachability preserver with size close to $OPT$?
This question is a version of \emph{Directed Steiner Network} (DSN), one of the most basic ``Steiner" problems in directed graphs, which form a central topic of study in combinatorial optimization; for more discussion, see the survey of Kortsarz and Nutov \cite{KN07}.


\begin{definition}[Directed Steiner Network]
Given a weighted directed graph $G=(V,E)$ and a set of demand pairs $P$, find the reachability preserver $H$ of minimum total weight $\sum_{e \in E(H)} w(e)$.
\end{definition}

(Of course, prior work does not use the language ``reachability preserver,'' but this formulation of the problem is equivalent to the usual one.)
Computation of sparse reachability preservers corresponds to the setting where the input graph is unweighted, called \emph{Unweighted Directed Steiner Network} (UDSN).
There is a long history of approximation algorithms for DSN and UDSN; see Table \ref{tab:dsnalgs}.
This has culminated in a relatively recent breakthrough by \Chlamtac{}, Dinitz, Kortsarz, and Laekhanukit \cite{CDKL17}, who achieved a better approximation factor of $O(n^{3/5+\eps})$ for UDSN.
That is, if the sparsest possible reachability preserver of the given input instance has OPT edges, then the algorithm of \Chlamtac{} et al.\ runs in polynomial time and finds a reachability preserver that has $O(n^{3/5+\eps} \cdot OPT)$ edges.
Our contribution is a rather simple application of Theorem~\ref{thm:main} to break beyond the $n^{3/5}$ bound achieved by \Chlamtac{} et al.

\begin{table*}[t]
\begin{center}
\begin{tabular}{lll}
    \toprule
    \textbf{Approx Ratio} & \textbf{Notes} & \textbf{Citation} \\
    \midrule
	$\tilde{O}(k^{2/3})$    & & Charikar et al.~\cite{CCCDGGL99}\\
	$\tilde{O}(k^{1/2+\eps})$ & & Chekuri, Even, Gupta, and Segev~\cite{CEGS11}\\
	$O(n^{4/5+\eps})$ & & Feldman, Kortsarz, and Nutov \cite{FKN12}\\
	$O(n^{2/3+\eps})$ & & Berman et al.~\cite{BBMRY13}\\
	$O(n^{3/5+\eps})$ & Unweighted only & \Chlamtac{}, Dinitz, Kortsarz, and Laekhanukit \cite{CDKL17}\\
	$O(n^{4/7 + \eps})$ & Unweighted only & \textbf{this paper}\\
	$\Omega\left(2^{\log^{1-\eps} n}\right)$ & Assuming NP $\ne$ QuasiP & Dodis and Khanna \cite{DK99}\\
    \bottomrule
    \end{tabular}
    
    \caption{\label{tab:dsnalgs} Approximation ratios for DSN that can be achieved by a polynomial-time algorithm.  Here $n$ is the number of nodes in the input graph, $k$ is the number of demand pairs, and $\eps$ can be any positive absolute constant (which trades off with the exponent in the polynomial runtime of the algorithm).}
    \end{center}
\end{table*}



\begin{theorem}
For all $\eps>0$, there is a polynomial time algorithm for UDSN with approximation ratio $O(n^{4/7 + \eps})$.
\end{theorem}

The previous algorithms use a subroutine that attempts to connect a pair set $P$ at a low cost, under the assumption that the pairs in $P$ have many paths between them (called ``thick'' pairs; the ``thin" pairs are handled separately, using Linear Programming).
To do this, these algorithms use the \emph{hitting set technique}: they randomly sample a small subset of the nodes $S$ that intersect at least one path for each pair in $P$, with high probability, and then they connect all nodes appearing in $P$ to-and-from each node in $S$.
All previous papers that follow this approach for Steiner-type problems (e.g. \cite{FKN12,BBMRY13,CDKL17,DZ16}) upper bound the number of edges contributed by this step as $O(n|S|)$.
Our improvement comes from applying the upper bound of Theorem~\ref{thm:main} instead.
Since the hitting set technique is ubiquitous in the design of graph algorithms, we believe that our general approach has potential for further application in approximation algorithms and beyond.
Our new approximation algorithm is probably not the final say on this fundamental problem; rather, it is a proof of concept that approximation algorithms can benefit from extremal results.
Notably, all previous progress on this problem \cite{FKN12,BBMRY13,CDKL17,DZ16} has come from improved LP rounding techniques, so these simple extremal question about reachability preservers provide a new way forward.

How far can our approach be pushed?
A natural bound to hope for, suggested by Feldman et al.~\cite{FKN12} is $O(\sqrt{n})$ approximation: this would match the algorithm of Gupta et al.\ for Steiner-Network in \emph{undirected} graphs \cite{GHNR10}, and undirected graphs seem better understood.
Our approach would get an $O(\sqrt{n})$ approximation for UDSN, if we can get a positive answer to the fundamental extremal question, which we address in the next subsection: \emph{Are linear size reachability preservers always possible?}

\subsection{Reachability Preserver Lower Bounds}

Recall that the upper bound of Theorem~\ref{thm:main} was $O(n+\sqrt{n |P| |S|})$.
Perhaps fewer edges are always sufficient?
Most optimistically:

\begin{center}
{\em Do all $n$-node graphs and demand pairs $P$ admit a reachability preserver on $O(n+|P|)$ edges?}
\end{center}

Note that this question is answered affirmatively in \emph{undirected} graphs, using any spanning forest.
For \emph{distance} preservers in undirected graphs, the possibility of such linear size distance preservers was refuted by Coppersmith and Elkin \cite{CE06}, and the construction for refutation has been crucial to the resolution of longstanding open questions in the area of graph spanners \cite{AB17jacm,ABP17}.

One can apply known \emph{layering} techniques to convert the distance preserver lower bounds by Coppersmith and Elkin \cite{CE06} to reachability preserver lower bounds.
This yields the following basic lower bound in the pairwise setting:
\begin{theorem}
For any positive integers $d, n, p$, there is an $n$-node graph and set of $|P|=p$ demand pairs for which any reachability preserver $H$ has size
$$|E(H)| = \Omega\left(n^{\frac{2}{d+1}} p^{\frac{d-1}{d}} \right).$$
\end{theorem}
We consider this lower bound to be a bit weak and unsatisfying.
To highlight a knowledge gap that it leaves: Theorem \ref{thm:intropairwise} implies that $O(n)$ edges always suffice for a reachability preserver for $|P| = O(n^{1/2})$ demand pairs.
Our lower bound (with $d=2$) shows that $\omega(n)$ edges are needed sometimes when $|P| = \omega(n^{2/3})$.
This leaves a polynomial gap in our understanding of the possibilities of linear-size reachability preservers: how many demand pairs admit an $O(n)$-size reachability preserver in general?
We have been unable to close this gap, and we consider this to be a central open problem in the area of reachability preservers.

However, we are able to obtain much more gratifying lower bounds in the source-restricted $P \subseteq S \times V$ setting.
We prove: 
\begin{theorem} \label{thm:main-lb-intro}
For any positive integers $n, p, \sigma$ satisfying $n^4 \sigma^6 \le p^9$, there is an $n$-node graph $G = (V, E)$ and set of $|P|=p$ demand pairs, with $P \subseteq S \times V$, $|S| = \sigma$, for which any reachability preserver $H$ has size
$$|E(H)| = \Omega\left( n^{4/5} p^{1/5} \sigma^{1/5} \right).$$
\end{theorem}

The most important consequence of this theorem is that, together with Theorem \ref{thm:main}, significantly advances knowledge on the settings in which $O(n)$ size source-restricted reachability preservers are generally available (see Figure \ref{fig:sourcewiselinear}).
That is, Theorem \ref{thm:main} implies that $O(n)$ edges suffice when $p\sigma = O(n)$, and Theorem \ref{thm:main-lb-intro} implies that $\omega(n)$ edges are sometimes needed when $p\sigma = \omega(n)$.
This holds anywhere in the parameter regime $\sigma = o(n^{1/3})$ (since $\sigma  > n^{1/3}$ and $\sigma p > n$ implies that $n^4 \sigma^6 > p^9$).
In a small range of parameters (with $\sigma = \omega(n^{1/3})$ and $p\sigma = \omega(n)$) the problem remains open.
See Figure \ref{fig:sourcewiselinear} for a visualization of these bounds. 

At a technical level, our construction for Theorem \ref{thm:main-lb-intro} uses a similar discrete geometry toolkit as in \cite{CE06}, but it uses a more careful analysis of the geometry of the construction in order to force a small number of source nodes.

\begin{figure} [ht] \centering
\begin{tikzpicture} [scale=0.8]

\node at (5, -2) {\Huge $|P|$};
\node at (-2, 5) {\Huge $|S|$};
\node at (-1, 0) {$1$};
\node at (-1, 10) {$n$};
\node at (0, -1) {$1$};
\node at (10, -1) {$n$};

\draw [fill=black!20] (0, 0) -- (0, 10) -- (10, 10) -- cycle;
\node  at (3, 8) {must have $|S| \le |P|$};

\draw [fill=green!50] (0, 0) -- (5, 5) -- (10, 0) -- cycle;

\draw (5, 0.1) -- (5, 5) -- (9.9, 0.1) -- cycle;
\draw [fill=red!50]  (10, 0) -- (6.66, 3.33) -- (6.66, 6.66) -- (10, 10) -- cycle;

\node at (6, 5) {\Huge ?};

\node [align=center] at (3, 1) {\Huge\checkmark{}\\(already known)};
\node [align=center] at (7, 1) {\Huge\checkmark{}\\(Theorem \ref{thm:main})};
\node [align=center] at (8.33, 5) {\Huge $\times$\\(Theorem \ref{thm:main-lb-intro})};

\node at (5, -0.7) {$n^{1/2}$};
\end{tikzpicture}

\caption{\label{fig:sourcewiselinear} Given an $n$-node graph $G = (V, E)$ and demand pairs $P \subseteq S \times V$, for which sizes of $|P|, |S|$ is there necessarily a reachability preserver on only $O(n)$ edges?  This diagram shows the state of knowledge following this paper: the answer is \emph{yes} when $|S||P| \le O(n)$, it is \emph{no} when $|S||P| = \omega(n)$ and $|P| = \Omega(n^{2/3})$, and it is open otherwise.  The regime in which a positive answer was previously known (in green) is implicit in the work of Coppersmith and Elkin \cite{CE06}, and also from the trivial observation that $O(1)$ source nodes have a reachability preserver on $O(n)$ edges via spanning trees.}
\end{figure}

An intriguing open question is to connect extremal results and approximation algorithms in another direction: can we use our new lower bound graphs to improve the inapproximability bounds for DSN?

\subsection{Related Work}

\emph{Distance preservers} are subgraphs that preserve distance between demand pairs, not just reachability.
Distance preservers have been extensively studied, and they work as a primitive for other problems in network design \cite{BCE05,CE06,BV16,Bodwin21sicomp,AB16soda,AB17jacm, GR17, CDKL17, CGMW18}.
Some prior work for distance preservers has direct implications for reachability preservers.
For example, the following theorem follows from an argument in prior work on distance preservers:
\begin{theorem} [\cite{Bodwin21sicomp}]
Let $\rs(n)$ be the largest value such that every graph $G = (V, E)$ whose edge set can be partitioned into $n$ induced matchings has $O\left(\frac{n^2}{\rs(n)}\right)$ edges.
Then all $G, P$ has a reachability preserver on $O(|P|)$ edges whenever $|P| = \Omega\left(\frac{n^2}{\rs(n)}\right)$.
\end{theorem}
It is known that $\rs(n)$ is a superconstant function; in particular, the current bounds are
$$2^{\Omega(\log^* n)} \le \rs(n) \le 2^{O(\sqrt{\log n})}$$
due to \cite{Fox11, Elkin10, Behrend46}.
Another theorem that follows directly from the same paper gives a lower bound in the ``subset'' setting:
\begin{theorem}
For all positive integers $n, d, \sigma$, there is an $n$-node graph and a set of $|S| = \sigma$ source nodes such that any reachability preserver $H$ of the demand pairs $S \times S$ has size
$$|E(H)| \ge n^{\frac{2}{d+1}} \sigma^{\frac{(2d+1)(d-1)}{d(d+1)} - o(1)}.$$
\end{theorem}
This theorem is proved for \emph{undirected distance} preservers in \cite{Bodwin21sicomp}, but the construction happens to use a layered graph.
It can thus be converted to a lower bound on reachability preservers by simply directing the edges down the layers.

We next discuss some other work more indirectly related to reachability preservers.
\emph{Pairwise spanners} are a relaxation in which distances between demand pairs only need to be preserved \emph{approximately} \cite{Woodruff06,CGK13,KV15,Parter14,Kavitha17}.
A \emph{distance preserving minor} is a small minor of $G$ that preservers all distances in $P$ approximately \cite{Gupta01,CXKR06,EEST08,BG08,EGKRTT14,KNZ14,CGH16,GR16,GHP20,KR20,Cheung18,Filtser19}.
Although pairwise spanners are most commonly studied for undirected graphs, there are other notions of sparsification more amenable to directed graphs.
A \emph{roundtrip spanner} is a sparse subgraph in which all pairwise \emph{roundtrip} distances ($\dist(u, v) + \dist(v, u)$) are approximately preserved \cite{CW04,RTZ08,PRSTV18, CDG20}.
Perhaps most related to ours are \emph{transitive closure spanners} \cite{BGJRW12,Raskhodnikova10}, which also focus on preserving reachability properties.
The goal there is to find a small graph (not necessarily subgraph) that has the same reachability relation as the input graph, and which has diameter as small as possible.

In the special case of $P=\{s\} \times V$, there has been exciting recent progress in the fault-tolerant setting \cite{PP13,Parter15,BCR16,Choudhary16, GK17, Parter15} which essentially studied the following question:
Given a graph $G$ and a source $s$, what is the sparsest subgraph $H$ such that for all nodes in $v$ there are at least $k$ node (or edge) disjoint paths in $H$ iff there are in $G$.
The questions we study are the special case of $k=1$, but we consider more than one source.
Recently, follow-up work to this paper by Chakraborty and Choudhary~\cite{CC20} studied general fault-tolerant reachability preservers in the general pairwise setting.


%
%

\section{Reachability Preserver Upper Bounds}
\label{sec:overview}

%
%

We begin by proving our extremal upper bounds on the size of reachability preservers, and then we give our algorithm implementing these bounds with fast running time.

\subsection{Proof of Theorem \ref{thm:main}}

Let $G = (V, E)$ be an $n$-node directed graph, let $S \subseteq V$, and let $P \subseteq S \times V$ be a set of demand pairs.
Our goal is to argue that a sparse reachability preserver of this input instance exists.
First off, for each strongly connected component of $G$, we can preserve reachability using a linear number of edges using an in- and out- reachability tree on this component.
This component may then be contracted into a single super-node, since we already have reachability among all pairs in the component.
Thus, by spending $O(n)$ edges in total, we can assume without loss of generality that $G$ is acyclic.
We make this assumption going forward.

Let $H = (V, E_H)$ be a reachability preserver of $G, P$ with the minimum possible number of edges $|E_H|$ (we will not yet worry about how to construct $H$ efficiently; we will just try to bound $|E_H|$).
\begin{definition} [Requirement and Collision]
For each demand pair $(s, t) \in P$, we say that $(s, t)$ \emph{requires} an edge $e \in E_H$ if every $s \leadsto t$ path in $H$ includes the edge $e$.
We will say that two demand pairs $p_1, p_2 \in P$ \emph{collide} on edges $e_1, e_2$ if:
\begin{itemize}
\item $p_1$ requires $e_1$,
\item $p_2$ requires $e_2$, and
\item $e_1, e_2$ are distinct, but have the same endpoint node.  (E.g., they have the form $e_1 = (u_1, v), e_2 = (u_2, v)$ for some node $v$.)
\end{itemize}
\end{definition}
We will next prove an intermediate claim on the structure of demand pair collisions.
In the following claim and throughout the paper, given a simple path $\pi(s, t)$ that contains two nodes $u, v$ in that order, we will use the notation $\pi(s, t)[u \leadsto v]$ to refer to the contiguous subpath of $\pi(s, t)$ starting at $u$ and ending at $v$.

%

\begin{claim} \label{clm:notriangle}
For any demand pair $p \in P$ and source node $s \in S$, there is at most one pair of edges $\{e_p, e_s\}$ with the property that there exists a demand pair $(s, t) \in P$ that uses $s$ as its start node and where $p, (s, t)$ collide on $\{e_p, e_s\}$.
\end{claim}
We note that Claim \ref{clm:notriangle} allows for the possibility that two demand pairs $(s, t), (s, t') \in P$ that both collide with $p$, as long as they do so via the same edge pair $\{e_p, e_s\}$.

\begin{proof} [Proof of Claim \ref{clm:notriangle}]
Suppose otherwise, towards a contradiction.
This means there are distinct edge pairs
$$\{e_p = (a, v), e_s = (b, v)\}, \{e'_p = (a', v'), e'_s = (b', v')\}$$
and demand pairs $(s, t), (s, t') \in P$ such that $p$ collides with $(s, t), (s, t')$ on these respective edge pairs.
Note again that $v, v'$ are distinct.

Let $\pi(p), \pi(s, t), \pi(s, t')$ be any fixed paths in $H$ for these demand pairs, and assume without loss of generality that $v$ precedes $v'$ along $\pi(p)$.
We can now generate an $s \leadsto t'$ path by concatenating subpaths of these three paths, as follows:
$$\pi(s, t)[s \leadsto v] \circ \pi(p)[v \leadsto v'] \circ \pi(s, t')[v' \leadsto t']$$
(see Figure \ref{fig:notriangle}).
Since $v \ne v'$, and since $p$ requires the edge $(a', v')$, we have $a' \in \pi(p)[v \leadsto v']$.
Thus, this concatenated path enters the node $v'$ along the edge $(a', v')$.
Since $a' \ne b'$, this means the concatenated path is an $s \leadsto t'$ path that does \emph{not} use the edge $(b', v')$.
This contradicts the hypothesis that the demand pair $(s, t')$ requires the edge $(b', v')$, completing the proof.
\end{proof}

\begin{figure}[ht]
\begin{center}
\begin{tikzpicture}
\draw [fill=black] (0, 0) circle [radius=0.15cm];
\node at (0, -0.4) {$s$};

\draw [fill=black] (-2, 2) circle [radius=0.15cm];
\node at (-1.65, 2.2) {$v$};

\draw [fill=black] (2, 2) circle [radius=0.15cm];
\node at (2.35, 2.3) {$v'$};

\draw [fill=black] (-2, 4) circle [radius=0.15cm];
\node at (-2, 4.4) {$t$};

\draw [fill=black] (2, 4) circle [radius=0.15cm];
\node at (2, 4.4) {$t'$};

\draw [->] (-3.5, 2) -- (3.5, 2);
\node at (4, 2) {$\pi(p)$};

\draw [fill=black] (-2.5, 2) circle [radius=0.15cm];
\node at (-2.8, 1.7) {$a$};

\draw [fill=black] (1.5, 2) circle [radius=0.15cm];
\node at (1.2, 1.7) {$a'$};

\draw [fill=black] (-2, 1.5) circle [radius=0.15cm];
\node at (-2, 1.1) {$b$};

\draw [fill=black] (2, 1.5) circle [radius=0.15cm];
\node at (2, 1.1) {$b'$};

\draw [->] plot coordinates {(0, 0) (-2, 1.5) (-2, 2) (-2, 3.8)};
\draw [->] plot coordinates {(0, 0) (2, 1.5) (2, 2) (2, 3.8)};
\node at (-2.7, 3) {$\pi(s, t)$};
\node at (2.7, 3) {$\pi(s, t')$};

\draw [->, red, line width = 3] plot coordinates {(0, 0) (-2, 1.5) (-2, 2) (2, 2) (2, 4)};

\end{tikzpicture}
\end{center}
\caption{\label{fig:notriangle} In the proof of Claim \ref{clm:notriangle}, we argue that in the pictured setup, the thick red path is an $s \leadsto t'$ path that avoids the edge $(b', v')$, contradicting that the edge $(b', v')$ is required by the demand pair $(s, t')$.}
\end{figure}

We now continue with our counting argument.
Let $d$ be a parameter of the argument, which is a positive integer.
We say that an edge $(u, v)$ is \emph{light} if $\deg_H(v) \le d$, or it is \emph{heavy} if $\deg_H(v) > d$.
Unioning over the $n$ nodes in $H$, there are $\le nd$ light edges in total.

To count the heavy edges, let $p \in P$ be an arbitrary demand pair, and suppose $p$ requires $h$ total heavy edges, where these edges enter nodes $\{v_1, \dots, v_h\}$.
Let $A$ be the set of all edges entering any of these nodes $\{v_1, \dots, v_h\}$ besides the edges required by $p$; we thus have $|A| \ge hd$.
Moreover, since $H$ is a reachability preserver of minimal size, none of the edges in $A$ may be removed from $H$ without disconnecting at least one demand pair.
So every edge $e \in A$ is required by a demand pair $q \ne p$, and $q, p$ collide via the edge $e$ (together with the appropriate edge in $p$).

By Claim \ref{clm:notriangle}, for each edge $s \in S$, there is at most one edge $e \in A$ for which there is a demand pair $(s, t) \in P$ that collides with $p$ via the edge $e$.
Thus we have
$$hd \le |A| \le |S|,$$
and so the number of heavy edges required by $p$ is $\le |S|/d$.
We may thus upper bound the total number of heavy edges in $H$ as
$$\sum \limits_{p \in P} \left|\{e \in E \ \mid \ e \text{ is heavy and required by } p \} \right| \le \frac{|P||S|}{d}.$$
Combined with the light edges, we have
$$|E_H| \le nd + \frac{|P||S|}{d}.$$
To complete the proof, we set
$$d = \left\lceil \left( \frac{|P||S|}{n} \right)^{1/2} \right\rceil,$$
which gives a bound of $|E_H| = O(\sqrt{n|P||S|})$ whenever the quantity $|P||S|/n \ge 1$, or (since we take a ceiling in the definition of $d$) a bound of $O(n)$ whenever $|P||S|/n < 1$, for a total bound of
$$|E(H)| = O\left( n + \sqrt{n|P||S|} \right).$$

\subsection{Proof of Theorem \ref{thm:intropairwise}}

Our goal is now to provide an upper bound for a reachability preserver of an $n$-node graph $G$ and an arbitrary set of demand pairs $P$ (not necessarily source-restricted).\footnote{The original version of this paper proved this theorem directly, as a consequence of Claim \ref{clm:notriangle}.
The following is a simplified version of the proof, derived from followup work by the second author \cite{Bodwin20}.}
For each demand pair $(s, t) \in P$, fix an arbitrary $s \leadsto t$ path $\pi(s, t)$ in $G$.
Let $\ell$ be a parameter, and let $R$ be a random sample of nodes, obtained by including each node independently with probability $\ell^{-1}$.
We say that a demand pair $(s, t)$ is \emph{hit} if we have sampled a node $r \in R$ in $\pi(s, t)$, or it is \emph{missed} otherwise.

\begin{itemize}
\item To handle the missed demand pairs $(s, t)$, we simply add all the edges of $\pi(s, t)$ to the reachability preserver.
Our goal is then to argue that, for each demand pair $(s, t)$, the expected number of edges it contributes in this case is $O(\ell)$.
This part of the argument is standard in the area (it follows e.g.\ from Chernoff bounds), and can be calculated as follows.
A path $\pi(s, t)$ is missed iff all $|\pi(s, t)|$ of its nodes are unsampled, and so
\begin{align*}
\Pr[(s, t) \text{ missed}] &= \left( 1 - \frac{1}{\ell} \right)^{|\pi(s, t)|}\\
&= \left(\left( 1 - \frac{1}{\ell}\right)^{\ell}\right)^{|\pi(s, t)|/\ell}\\
&= c^{|\pi(s, t)|/\ell} \tag*{for some $0 < c < 1$.}
\end{align*}

Thus, the expected number of edges contributed by a demand pair $(s, t)$ due to the event that it is missed is either $O(\ell)$ in the case where $|\pi(s, t)| = O(\ell)$, or it is exponentially decaying in the quantity $|\pi(s, t)|/\ell$ otherwise.
That is, we may compute:
\begin{align*}
\mathbb{E}[|\pi(s, t)| \cdot \mathbb{I}[(s, t) \text{ missed}]] &\le O(\ell) + |\pi(s, t)| \cdot c^{|\pi(s, t)|/\ell}\\
&= O(\ell) + O(\ell) \cdot \left(\frac{|\pi(s, t)|}{\ell}\right) \cdot c^{|\pi(s, t)| / \ell}\\
&= O(\ell) + O(\ell) \cdot O(1) \tag*{since $c<1$}\\
&= O(\ell).
\end{align*}
Finally, unioning over all demand pairs, we get $O(|P| \ell)$ edges in expectation from missed demand pairs.

\item To handle the hit demand pairs $(s, t)$, we first split them into two demand pairs $(s, r), (r, t)$ where $r \in R$ (clearly, if we preserve $s \leadsto r$ and $r \leadsto t$ reachability, then we also implicitly preserve $s \leadsto t$ reachability).
Let $P_1$ denote the set of the first of two demand pairs arising from a hit demand pair (e.g., $(s, r) \in P_1$), and let $P_2$ denote the set of the second of these two demand pairs (e.g., $(r, t) \in P_2$).
We first measure the cost of preserving $P_2$.
We have the structure $P_2 \subseteq R \times V$, so applying Theorem \ref{thm:main} we have a reachability preserver for $P_2$ at cost
$$O\left(n + \sqrt{n\left|P_2\right||R|}\right).$$
To bound the expected size of this reachability preserver, we compute 
\begin{align*}
\mathbb{E}\left[ n + \sqrt{n |P_2| |R|} \right] &\le \mathbb{E}\left[ n + \sqrt{n |P| |R|} \right]\\
&=  n + (n |P|)^{1/2} \cdot \mathbb{E}\left[|R|^{1/2}\right]\\
&\le  n + (n |P|)^{1/2} \cdot \mathbb{E}\left[|R|\right]^{1/2} \tag*{Jensen's Inequality}\\
&= n + (n |P|)^{1/2} \cdot \left( n\ell^{-1} \right)^{1/2}\\
&= n + n \sqrt{\frac{|P|}{\ell}}.
\end{align*}
By symmetric logic, the expected cost to preserve demand pairs in $P_1$ is the same.\footnote{This requires an application of Theorem \ref{thm:main} with demand pairs of the form $P \subseteq V \times S$, rather than $P \subseteq S \times V$ as in the statement of Theorem \ref{thm:main}.  However, these settings are clearly symmetric and so the same reachability preserver bounds apply, e.g., by reversing the edges of the input graph.}
\end{itemize}

So the total expected number of edges in this reachability preserver construction is
$$O\left(n + |P| \ell + n \sqrt{\frac{|P|}{\ell}} \right).$$
We now balance parameters by setting $\ell :=  n^{2/3} |P|^{-1/3}$, which gives a total expected cost of
\begin{align*}
O\left(n + |P| \cdot n^{2/3} |P|^{-1/3} + n \sqrt{\frac{|P|}{n^{2/3} |P|^{-1/3}}}\right) &= O\left(n + n^{2/3} |P|^{2/3} + n^{2/3} |P|^{2/3}\right)\\
&= O\left(n + n^{2/3} |P|^{2/3} \right)
\end{align*}
for our preserver, as claimed.

\subsection{Constructing Reachability Preservers \label{sec:fastalg}}

Here we observe that one can construct asymptotically existentially optimal reachability preservers in $O(|E| \cdot |S| \log n)$ time.
By ``existentially optimal,'' we mean the following.
Let $\rpp(n, p, \sigma)$ be the smallest integer such that, for any $n$-node directed graph $G = (V, E)$, set of $|S| = \sigma$ source nodes, and set of $|P| = p$ demand pairs $P \subseteq S \times V$, there is a reachability preserver on $\le \rpp(n, p, \sigma)$ edges.
So for example, Theorem \ref{thm:main} can be equivalently phrased as the statement
$$\rpp(n, p, \sigma) = O\left(n + \sqrt{np\sigma}\right).$$

An existentially optimal algorithm is one that produces reachability preservers on $O(\rpp(n, p, \sigma))$ edges, whatever this function value may be.
So the number of edges is \emph{at most} $O(n + \sqrt{np\sigma})$, but it could be substantially smaller if the upper bound of Theorem \ref{thm:main} winds up being significantly suboptimal.
That is:

\begin{theorem} \label{thm:reach-construct}
There is a randomized algorithm that, given an $n$-node directed graph $G = (V, E)$, source nodes $S$, and demand pairs $P \subseteq S \times V$ on input, constructs a reachability preserver on $O(\rpp(n, |P|, |S|))$ edges (always), and terminates in time $O\left(|E| |S| \log n)\right)$ with high probability.
\end{theorem}

We comment that the following approach also works to produce reachability preservers of existentially optimal size with respect to a parametrization based only on $n$ and $p$ (but not $\sigma$), like in Theorem \ref{thm:intropairwise}.
That is: if $\rpp(n, p)$ denotes the least integer such that every $n$-node graph and set of $p$ demand pairs has a reachability preserver on $\le \rpp(n, p)$ edges, the following algorithm always produces a reachability preserver on $O(\rpp(n, p))$ edges.
However, the runtime of the algorithm will still depend on $\sigma$, the number of source nodes used by the demand pairs, even if the size bounds do not.


\paragraph{Step 1: Convert $G$ to a DAG.}

Like we did in our previous extremal proof, it will be helpful to first convert $G$ to a DAG.
To accomplish this, we run an algorithm to detect the strongly connected components of $G$ in $O(|E|)$ time (e.g. \cite{Tarjan72}), and then we spend $O(n)$ edges to add in- and out- reachability trees that preserve all-pairs reachability within each strongly connected component.
We may then contract each strongly connected component into a single vertex, and build a reachability preserver of the remaining graph, which is acyclic.
At the end of the algorithm, we would then un-contract each strongly connected component, and replace each reachability preserver edge $(u, v)$ with any single edge going from the component represented by the node $u$ to the component represented by the node $v$.

One can see that $\rpp(n, p, \sigma) = \Omega(n)$, since there are inputs where $\Omega(n)$ edges are needed for a reachability preserver.
For example, if the input graph is a path and $|P|$ contains the pair of nodes at either extreme end of the path, then any reachability preserver must keep all $n-1$ edges.
Thus, the $O(n)$ cost of preserving reachability within strongly connected components is $O(\rpp(n, |P|, |S|))$, so it can be ignored.
In the following, we assume that $G$ is a DAG.

\paragraph{Step 2: Building the Reachability Preserver.}

We will construct our reachability preserver of $G, P$ decrementally; that is, we initially set $H \gets G$ and we will iteratively delete edges of $H$.
We use as a subroutine an algorithm of Italiano \cite{Italiano88}.

\begin{theorem} [\cite{Italiano88}] \label{thm:dec-reach}
There is a deterministic algorithm that, given a DAG $G = (V, E)$ and a source node $s \in V$, explicitly maintains the set of nodes reachable from $s$ over a sequence of edge deletions.
The total amount of time needed to maintain this list over all edge deletions is $O(|E|)$.
\end{theorem}

For the sake of building intuition, we first consider Algorithm \ref{alg:warmup}, which is perhaps the most natural method for sparsifying $H$ while preserving reachability among pairs in $P$ (this is \textbf{not} the final algorithm that we use).
\begin{algorithm} 
Initialize $H \gets G$\;

\ForEach{$s \in S$}{
Initialize a data structure $D_s$ as in Theorem \ref{thm:dec-reach}\;
}

\While{$H$ has $\ge 2\rpp(n, |P|, |S|)$ edges remaining}{
Choose an edge $e$ still in $H$ uniformly at random\;
Delete $e$ from $H$ and update each data structure $D_s$ accordingly\;
\If{any demand pair $(s, t) \in P$ is no longer reachable in $H$}{
Add $e$ back to $H$ and undo the changes made to all data structures $D_s$\;
}
}
\Return{$H$}\;
\caption{\label{alg:warmup} Warmup Algorithm for Constructing Reachability Preservers}
\end{algorithm}

It is immediate that Algorithm \ref{alg:warmup} is correct, in the sense that it eventually returns a reachability preserver with $O(\rpp(n, |P|, |S|))$ edges.
The trouble is that its runtime guarantees are not very good.
Let us say that an iteration of the main while loop is \emph{successful} if the selected edge $e$ does not affect reachability among demand pairs, and so $e$ remains deleted.
The successful iterations are not a problem: using Theorem \ref{thm:dec-reach}, they take $O(|E||S|)$ time in total.

The problem is that we expect to have $\Omega(\rpp(n, |P|, |S|))$ unsuccessful iterations, and \emph{each} unsuccessful iteration might take $\Omega(|E||S|)$ time.
That is, because the \emph{worst case update time per deletion} in Theorem \ref{thm:dec-reach} is $O(|E|)$ for each data structure, it is  conceivable that we will pay $\Omega(|E||S|)$ work for a single unsuccessful deletion, and then we have to unwind all of this work and so we are not able to amortize it over the runtime of the entire algorithm.

This failed attempt at an algorithm gives us the intuition that we are willing to perform some extra work in order to avoid unsuccessful iterations.
The key insight here is that \emph{parallelization} is useful.
In particular, our final algorithm (Algorithm \ref{alg:reach}) works by maintaining $\Theta(\log n)$ different ``universes'' at a time, and it runs each loop through Algorithm \ref{alg:warmup} \emph{simultaneously in all universes}.
This makes it extremely likely that the iteration will be successful in at least one universe, which is all we need to get our runtime bounds.

\begin{algorithm}[ht] 
Initialize $H \gets G$\;

\ForEach{each source $s \in S$}{
Initialize $c \log n$ identical data structures $D^i_s$ from Theorem \ref{thm:dec-reach} ($i \in [c \log n]$)\;
}

\While{$H$ has $\ge 2 \rpp(n, |P|, |S|)$ edges remaining}{
Let $R$ be a uniform random sample of $c \log n$ edges still in $H$\;
\ForEach{edge $r_i \in R$ \textbf{in parallel:}}{
Update the data structures $D_s^i$ with the deletion of $r_i$ for each $s \in S$\;
\If{all pairs in $P$ are still reachable after $r_i$ is deleted}{
Delete $r_i$ from $H$\;
Halt the parallel process for each other edge $r \in R$\;
Undo the updates to all other data structures $\{D_s^{j \ne i}\}$ in this iteration\;
Update all other data structures $D_s^j$ by deleting $r_i$\;
}
}
}
\Return{$H$}\;
\caption{\label{alg:reach} Fast Construction of Reachability Preservers}
\end{algorithm}

Just like before, it is easy to see that this algorithm produces a correct reachability preserver on $O(\rpp(n, |P|, |S|))$ edges: the edge bound follows from the stopping condition on the main while loop, and correctness follows from the fact that we only delete an edge $r_i$ from $H$ when doing so does not destroy reachability among any demand pairs.

It will be a crucial detail in our analysis that the inner for loop is executed \emph{in parallel} over the sampled edges $r_i \in R$; that is, we spend one computational step to progress the updates to data structures $\{D_s^i\}$ before moving on to the next edge $r_{i+1} \in R$ and spending one computational step to progress data structures $\{D_s^{i+1}\}$, and so on.
In an iteration of the main while loop, let us say that a particular sampled edge $r_i \in R$ is \emph{successful} if all pairs in $P$ are still reachable in $H \setminus \{r_i\}$ (there can be many successful edges in each iteration).
In particular, the parallelization means that if the data structure updates for a successful edge $r_i$ require $t$ steps, then our algorithm spends only $O(t \log n)$ steps on this iteration in total, since it progresses all $c \log n$ universes by $t$ steps each before halting due to the completed updates for $r_i$.

The constant $c$ in the algorithm can be any constant greater than $2$, whose value governs the probability that the algorithm halts within our claimed runtime.
In particular:

\begin{claim}
If $c > 2$, then with high probability, in every round of the main while loop at least one sampled edge $r \in R$ is successful.
\end{claim}
\begin{proof}
Since $H$ has $\ge 2 \rpp(n, |P|, |S|)$ edges, at least half of the edges in $H$ can be deleted without destroying reachability for any demand pairs.
So for any single sampled edge $r \in R$, we have
$\Pr\left[ r \text{ successful} \right] \ge 1/2$.
It follows that
$$\Pr\left[ \text{there exists successful } r \in R \right] \ge 1 - 1/2^{c \log n} = 1 - 1/n^c.$$
By an intersection bound, we then have
$$\Pr\left[ \text{in all of the first } n^2 \text{ iterations, there exists successful } r \in R \right] \ge 1 - 1/n^{c-2}.$$
Since $H$ initially has $\le n^2$ edges, and an edge is deleted in each iteration of the while loop for which at least one sampled edge $r \in R$ is successful, this condition suffices for every iteration to be successful.
\end{proof}

So in the following, we will assume that at least one sampled edge $r \in R$ is successful in each round, and we will say that the \emph{first successful edge} $r \in R$ is the one for which the data structure updates terminate the fastest.
Thus, among all successful edges in $R$, the \emph{first} successful edge is the only one that ultimately gets deleted from $H$ in this iteration.
Let $r^j$ be the first successful edge in iteration $j$, and let $t^{(j)}$ be the time required to update the data structures associated to $r^j$ in this iteration.
Thus, as discussed previously, the $j^{th}$ iteration of the algorithm runs in $O(t^{(j)} \log n)$ time, and the total runtime of the algorithm is
$$\sum \limits_j O\left(t^{(j)} \log n\right) = O\left(\log n \right) \cdot \sum \limits_j t^{(j)}.$$
To bound this inner summation, it is helpful to briefly imagine that we only have $|S|$ data structures $\{D_s\}$, and we update these data structures by deleting the first successful edges $\{r^1, r^2, \dots\}$ in sequence.
The updates associated to the deletion of $r^j$ require exactly $t^j$ time.
Moreover, using Theorem \ref{thm:dec-reach}, the total time required is $O(|E|)$ per data structure, for a total of $O(|S||E|)$ across all data structures.
We thus have
$$\sum \limits_j t^{(j)} = O(|S||E|).$$
So the total runtime of Algorithm \ref{alg:reach} is $O(|S||E| \log n)$, as claimed.

\section{Applications to Directed Steiner Network}
\label{sec:approx}

In this section we obtain a new approximation algorithm for Unweighted Directed Steiner Network (UDSN).
Our algorithm builds on prior work by identifying an ingredient that is common to most previous approaches, and we show how it can benefit from our extremal results on reachability preservers.

Let us fist briefly review the state-of-the-art algorithm of Chlamatac et al.~\cite{CDKL17} for UDSN.
This algorithm really contains two different algorithms; one guarantees a factor $k = n^{3/5 + \eps}$ approximation in the setting where $OPT \le n^{4/5}$, and the other achieves the same approximation factor in the setting where $OPT \ge n^{4/5}$.
(Throughout this exposition, $\eps$ can be any positive constant, which trades off with the exponent in the polynomial-time algorithm.)
In a little more detail, the first algorithm in \cite{CDKL17} yields the following result:

\begin{lemma}[follows from \cite{CDKL17,BBMRY13}] \label{lem:lowoptdsn}
If a UDSN instance has $OPT \leq O(n^{4/5-\alpha})$ for some $\alpha \ge 0$, then there is a polynomial time algorithm that gives a $k \leq O(n^{3/5-\alpha/3+\eps})$ approximation to $OPT$.
\end{lemma} 

We will use this lemma exactly as stated here, with no changes to the underlying algorithm.
Rather, our improvements apply to the second algorithm, which applies for larger values of $OPT$.
This algorithm is based on a dichotomy between \emph{thick} and \emph{thin} demand pairs.
We pick a threshold $k$, which is a parameter that we will select later, and say that a demand pair $(s,t) \in P$ is $k$-\emph{thick} if the set of all $s \leadsto t$ paths in $G$ contains at least $k$ distinct nodes, and otherwise the demand pair $(s, t)$ is $k$-\emph{thin}.
The thin pairs are again handled using a subroutine from prior work.
The proof of the following lemma relies on an LP relaxation of the problem, and then a clever randomized rounding strategy to pick an approximate integral solution. 
\begin{lemma}[follows from \cite{BBMRY13}, used in \cite{CDKL17}] \label{lem:thinpairs}
For all $k \geq 1$, given an instance of UDSN we can find a subgraph on $\Oish(k \cdot OPT)$ edges, in which all $k$-thin pairs are connected with high probability.
\end{lemma}

So the thin demand pairs can be handled within approximation ratio $k$.
We now turn to the thick pairs.
All previous works for DSN and related problems \cite{FKN12,BBMRY13,CDKL17,DZ16}, where this thin/thick pairs framework was used, handled the thick pairs using a naive strategy: they sample a hitting set $S$ of $\Oish(n/k)$ nodes, arguing that $S$ contains a node along an $s \leadsto t$ path for every thick pair $(s, t)$.
Then they try to connect every terminal in the pair set $P$ to every node in the hitting set $S$.
For instance, Chlamatac et al.\ take BFS trees in and out of each node in the hitting set. In their algorithm, $k$ is set to $n^{3/5}$ and so their hitting set has size $\Oish(n^{2/5})$, which makes the cost of this stage $\Oish(n^{7/5})$.

But do we really need $O(n^{7/5})$ edges in order to connect all the terminals to the hitting set?
This is where our work comes in: Theorem \ref{thm:main} exactly implies that we can do much better.
For example, say that OPT is $n^{4/5}$ and that we have $n^{4/5}$ terminals that we want to connect to $n^{2/5}$ other nodes. 
Theorem~\ref{thm:main} says that $O(n^{13/10})$ edges suffice, improving on the naive bound of $n^{14/10}$.



More concretely, let $S$ be a hitting set of size
$$|S| = \Oish\left(n/k\right)$$
(that is, for every $k$-thick pair $(s, t)$, there exists an $s \leadsto t$ path that contains a node $x \in S$).
Let $T$ be the set of all terminals participating in $k$-thick pairs in $P$.
Notice that $|T|\leq OPT$, since any solution must keep at least one edge adjacent to each terminal in $P$.
Now our goal is to connect all nodes in $T$ to and from all nodes in $S$; that is, we consider the pair sets $P_1 := S \times T, P_2 := T \times S$, and we build reachability preservers in $G$ for $P_1, P_2$.
Theorem~\ref{thm:main} implies that the total number of edges needed for these reachability preservers is
\begin{align*}
&O\left(n + \sqrt{n |S|^2 |T|}\right)\\
=& \Oish \left(n + \sqrt{n \cdot \left(\frac{n}{k} \right)^2 \cdot OPT }\right) \\
=& \Oish \left(n + \frac{n^{3/2}}{k} \cdot \sqrt{OPT}\right).
\end{align*}

Let's now assume that $OPT \ge \Omega(n^{4/5 - \alpha})$, since this is the remaining case from Lemma \ref{lem:lowoptdsn}.
The approximation ratio obtained is thus
\begin{align*}
=& \Oish \left( \frac{n + \frac{n^{3/2}}{k} \cdot \sqrt{OPT}}{OPT} \right)\\
=& \Oish \left( \frac{n}{OPT} + \frac{n^{3/2}}{k \sqrt{OPT}} \right)\\
=& \Oish \left( n^{1/5 + \alpha} + \frac{n^{11/10 + \alpha/2}}{k} \right).
\end{align*}

We would like the thick and thin pairs to incur the same approximation ratio.
A parameter balance gives that this occurs when we set
$$k := n^{11/20 + \alpha/4},$$
in which case the approximation ratio for thick pairs becomes
$$\Oish \left( n^{1/5 + \alpha} + n^{11/20 + \alpha/4} \right) = \Oish \left( n^{1/5 + \alpha} + k \right).$$
In the range $0 \le \alpha \le 3/5$, the latter term dominates and the approximation ratio is $\Oish(k)$.
Combined with Lemma \ref{lem:thinpairs}, this gives:

\begin{lemma}[new]
If a UDSN instance has $OPT \geq \Omega(n^{4/5-\alpha})$ for some $0 \le \alpha \le 3/5$, then there is a polynomial time algorithm that gives a $k \leq \Oish(n^{11/20 + \alpha/4})$ approximation to OPT.
\end{lemma}

We now have two UDSN algorithms which have approximation ratios of
$$n^{3/5 - \alpha/3 + \eps}, n^{11/20 + \alpha/4 + \eps}$$
respectively.
We can thus run both algorithms on a given input and take the sparser of the two solutions.
The two approximation ratios are equal when $\alpha = 3/35$, and both are $n^{4/7 + \eps}$; since they depend oppositely on $\alpha$, for any other choice of $\alpha$ one algorithm or the other beats the approximation ratio of $n^{4/7 + \eps}$.
This gives:

\begin{theorem}
For any fixed constant $\eps>0$, there is a polynomial time algorithm for UDSN with approximation factor $O(n^{4/7 + \eps})$.
\end{theorem}

\section{Reachability Preserver Lower Bounds \label{sec:extremal}}
In this section we supply extremal lower bounds for reachability preservers: that is, we construct particular input graphs and sets of demand pairs for which no sparse reachability preserver exists.
These provide limits to the general upper bounds that can be proved, along the lines of Theorems \ref{thm:main} and \ref{thm:intropairwise}.
In order to phrase these theorems, we will reuse notation from Section \ref{sec:fastalg}: let $\rpp(n, p)$ denote the smallest integer such that every $n$-node graph and set of $p$ demand pairs has a reachability preserver on $\le \rpp(n, p)$ edges, and define $\rpp(n, p, \sigma)$ similarly with the additional constraint that the demand pairs have the form $P \subseteq S \times V$, for a node subset of size $|S| = \sigma$.

\subsection{Pairwise Lower Bounds}

We begin by providing lower bounds against $\rpp(n, p)$.
Our starting point is the following theorem from prior work:

\begin{theorem} [Proved in \cite{CE06}] \label{thm:pairwise-distance-lb}
For any positive integers $d, n, p$, there is an $n$-node undirected unweighted graph $G = (V, E)$ with
$$|E| = \Omega\left(n^{\frac{2d}{d^2 + 1}} p^{\frac{d^2 - d}{d^2 + 1}}\right),$$
and a set of $|P|=p$ demand pairs such that
\begin{itemize}
\item For each pair $(s, t) \in P$ there is a unique shortest $s \leadsto t$ path in $G$,
\item These unique shortest paths are pairwise edge disjoint, and
\item The edge set of $G$ is precisely the union of these paths.
\end{itemize}
\end{theorem}

It will be convenient to add one more property to Theorem \ref{thm:pairwise-distance-lb}: that all unique shortest paths have exactly the same length $\ell$.
To enforce this, let us define
$$\ell := \left\lfloor \frac{|E|}{2p} \right\rfloor = \Theta\left(n^{\frac{2d}{d^2 + 1}} p^{\frac{-d-1}{d^2+1}} \right)$$
so $\ell$ is half the average length of a shortest path for one of the demand pairs (rounded down).
Then, for each demand pair $(s, t)$ with unique shortest path $\pi(s, t)$, we partition $\pi(s, t)$ into subpaths of length exactly $\ell$ each, plus a ``remainder'' path at the end which may be shorter:
$$\pi(s, t) = \underbrace{\pi(s=x_0, x_1) \circ \dots \circ \pi(s, x_k)}_{\text{length } \ell} \circ \underbrace{\pi(x_k, t)}_{\text{length in } [0, \ell-1]}.$$
We then replace the demand pair $(s, t)$ with the set of all demand pairs $\{(x_i, x_{i+1})\}$, we discard the final remainder pair $(x_k, t)$, and we remove all edges in $\pi(x_k, t)$ from $G$.
It is now clear that we have unique edge-disjoint shortest paths for our demand pairs, and they all have length exactly $\ell$.
The total number of edges discarded is $\le p \ell \le |E|/2$, which only affects the number of edges in $G$ by a constant factor.
Finally, since the number of edges $|E|$ changes by a constant factor, and the average path length changes by a constant factor (from about $2\ell$ to $\ell$), it follows that the number of demand pairs only changes by a constant factor as well.
All of these constant-factor changes affect only implicit constants in the relevant statistics for Theorem \ref{thm:pairwise-distance-lb}, and may be ignored.

We can now prove our lower bound for reachability preservers:
\begin{theorem} \label{thm:unconditional-pairwise-lb}
For any positive integer $d$, we have $ \rpp(n, p) = \Omega\left(n^{\frac{2}{d+1}} p^{\frac{d-1}{d}} \right).$
\end{theorem}

To prove Theorem \ref{thm:unconditional-pairwise-lb}, we will take a graph from Theorem \ref{thm:pairwise-distance-lb} and convert it to a reachability preserver lower bound by ``layering.''
That is, letting $G = (V, E)$ be an instance from Theorem \ref{thm:pairwise-distance-lb}, modified as above, we define a graph $G^* = (V^*, E^*)$ and demand pairs $P^*$ by the following process:
\begin{itemize}
\item The nodes $V^*$ are formed by $2\ell + 1$ distinct copies of $V$, labeled $\{V_0, \dots, V_{2\ell}\}$.
The nodes in $V_i$ are called the \emph{$i^{th}$ layer} of $G^*$.
For a node $v \in V$, we write $v_i$ to mean the copy of $v$ in the $i^{th}$ layer.

\item For each undirected edge $(u, v) \in E$, we include directed edges $(u_i, v_{i+1})$ and $(v_i, u_{i+1})$ in $E^*$ for all $0 \le i < 2\ell$.

\item For each demand pair $(s, t) \in P$, we include demand pairs $(s_i, t_{i+\ell})$ in $P^*$ for all $0 \le i \le \ell$.
\end{itemize}

We next claim that the demand pairs in $P^*$ all have unique edge-disjoint paths of length $\ell$ each in $G^*$.
To see this, notice that for any demand pair $(s_i, t_{i + \ell}) \in P^*$, by construction there is a unique $s_i \leadsto t_{i+\ell}$ path of length $\le \ell$ in $G^*$, since there is a unique shortest $s \leadsto t$ path in $G$ which has length $\ell$.
Moreover, there is no $s_i \leadsto t_{i+\ell}$ path of length $\ge \ell+1$, since any such path would terminate at layer $i + \ell + 1$ or beyond.

We now count the statistics of our reachability preserver lower bound.
The number of nodes in $G^*$ is
$$|V^*| =: \overline{n} = \Theta(n \ell) = \Theta\left(n^{\frac{d^2 + 2d + 1}{d^2 + 1}} p^{\frac{-d-1}{d^2+1}} \right) = \Theta\left(n^{\frac{(d+1)^2}{d^2 + 1}} p^{\frac{-d-1}{d^2+1}} \right).$$
The number of demand pairs in $P^*$ is
$$|P^*| =: \overline{p} = \Theta(p \ell) = \Theta\left(n^{\frac{2d}{d^2 + 1}} p^{\frac{d^2 - d}{d^2+1}} \right).$$
Any reachability preserver of $G^*, P^*$ must keep $\ell$ edges per demand pair, since there are unique edge-disjoint paths of length exactly $\ell$ for the demand pairs.
Thus the number of edges required for a reachability preserver $H$ is at least
\begin{align*}
|E(H)| \ge |P^*| \ell &= \Theta\left(n^{\frac{2d}{d^2 + 1}} p^{\frac{d^2 - d}{d^2+1}} \right) \cdot \Theta\left(n^{\frac{2d}{d^2 + 1}} p^{\frac{-d-1}{d^2+1}} \right)\\
&= \Theta\left( n^{\frac{4d}{d^2 + 1}} p^{\frac{d^2 - 2d - 1}{d^2 + 1}} \right).
\end{align*}
We want to phrase this edge lower bound in terms of $\overline{n}$ and $\overline{p}$, rather than $n$ and $p$.
This is a matter of straightforward algebra:
\begin{align*}
|E(H)| &\ge \Theta\left(n^{\frac{4d}{d^2 + 1}} p^{\frac{d^2 - 2d - 1}{d^2+1}} \right)\\
&= \Theta\left(n^{\frac{2d+2}{d^2 + 1}} p^{\frac{-2}{d^2+1}} \right) \cdot \Theta\left(n^{\frac{2d-2}{d^2 + 1}} p^{\frac{d^2 - 2d + 1}{d^2+1}} \right)\\
&= \Theta\left(n^{\frac{(d+1)^2}{d^2 + 1}} p^{\frac{-d-1}{d^2+1}} \right)^{\frac{2}{d+1}} \Theta\left(n^{\frac{2d}{d^2 + 1}} p^{\frac{d^2 - d}{d^2+1}} \right)^{\frac{d-1}{d}}\\
&= \Theta\left(\overline{n}^{\frac{2}{d+1}} \overline{p}^{\frac{d-1}{d}}\right),
\end{align*}
which proves Theorem \ref{thm:unconditional-pairwise-lb}.

\subsection{Lower Bounds in the $P \subseteq S \times V$ Setting}

We now prove lower bounds on $\rpp(n, p, \sigma)$.
The parameters of the construction are $r, w, h, \ell$, which are all positive integers, and which respectively stand for \emph{radius}, \emph{width}, \emph{height}, \emph{layers}.
They will satisfy the inequalities $3 \ell r < h \le w$.

\paragraph{Initial Pairwise Lower Bound.}

We start with $n := w h (\ell+1)$ nodes in our graph; specifically, the vertex set is
$$V = \zz_h \times \zz_w \times \{0, \dots, \ell\}.$$

The first two indices are the integers mod $h$ and the integers mod $w$; we will perform addition in these indices, and this is interpreted as modular arithmetic with rollover.
The last index is just an integer from $0$ to $\ell$ and we will never add to this index in such a way that rollover occurs.
We say that a node is \emph{in layer $i$} if its last index is $i$.
To avoid confusion with edges or demand pairs, we will often write nodes using bracket notation, for example $[v, i]$ where $v \in \zz_h \times \zz_w$ and $0 \le i \le \ell$.

Our next step is to define a set of \emph{critical vectors} $C$.
These are drawn from the following result in prior work:
\begin{theorem} [\cite{BL98}]
There exists a convex set of vectors $C$ in the integer lattice $\zz^2$ of size 
$|C| = \Theta \left( r^{2/3} \right)$
such that all vectors $c \in C$ have length $\le \|r\|_2$, and all vectors lie in the cone formed by positive linear combinations of the vectors $(2, 1)$ and $(1, 2)$.\footnote{Technically \cite{BL98} constructs a set of vectors that do not lie in this restricted cone, but their argument immediately implies that we may have a constant fraction of the vectors in any cone of constant angular width.}
\end{theorem}

By ``convex'' in this theorem, we mean the following specific property: for any distinct vectors $c_1, c_2 \in C$, the projection of $c_1$ in the direction of $c_2$ is shorter than $c_2$ itself.
That is:
$$\left\langle c_1, \frac{c_2}{\|c_2\|_2} \right\rangle < \|c_2\|_2.$$

We will now use these critical vectors to define a set of \emph{critical paths} $\Pi$ that will be important in our analysis.
There are $|\Pi| = hw |C| = \Theta( hwr^{2/3})$ total critical paths.
Specifically: for each node $[v, 0]$ in the $0^{th}$ layer and for each vector $c \in C$, we include the $[v,0] \leadsto [v+\ell c, \ell]$ path that uses the node $[v + ic, i]$ in layer $i$.
We will write this critical path using the shorthand $\pi_{v, c}$.
Our edge set is exactly the set of edges that appear in any critical path.

Let us pause to comment on the construction so far.
Following a proof strategy from \cite{CE06}, one can prove that each critical path is the unique path between its endpoints (we will do so as part of Lemma \ref{lem:crituniq} to follow).
Given this, if we take the endpoints of our critical paths to be our demand pairs, then any reachability preserver must keep all edges in the graph.
Setting $h=w$ then exactly recovers the previous pairwise lower bound in Theorem \ref{thm:unconditional-pairwise-lb} for $d=2$.
Setting $h \ll w$ yields a strictly weaker pairwise lower bound than the pairwise $d=2$ one; the tradeoff is that when $h$ is small we don't have to introduce quite as many source nodes in order to cover our demand pairs.
We overview this process of introducing source nodes next.

\paragraph{Similar Paths and Intuition for the Rest of the Argument.}

Now that we have a good pairwise lower bound construction, our next step is to augment the construction by adding a few additional source nodes $S$, and for each critical path, we will carefully select one node from $S$ and add it to the beginning of the path (giving \emph{extended critical paths}).
There is a danger here: if we add a random or arbitrary node from $S$ to the beginning of each extended critical path, then we might destroy path uniqueness.
Imagine that $\pi_1 = (x_0, \dots, x_k)$ and $\pi_2 = (y_0, \dots, y_k)$ are both critical paths, and hence the unique path between their endpoints.
Imagine that we add a node $s \in S$ to the beginning of both paths, giving $\pi'_1 = (s, x_0, \dots, x_k), \pi'_2 = (s, y_0, \dots, y_k)$.
These remain the unique paths between their endpoints iff there were initially no $x_0 \leadsto y_k$ or $y_0 \leadsto x_k$ paths, which is not guaranteed for all critical paths in our initial construction.

So, our strategy will be to partition our critical paths into a small number of subsets, where within each subset all pairs of paths are \emph{similar} to each other.
Intuitively, if two paths are similar, then they have a certain geometric structure that the above example does not happen, and we \emph{can} safely augment this pair with the same source node from $s$.
Thus, we can safely cover all of our critical paths using one source node per part in the partition.
The formal definition of path similarity is as follows:

\begin{definition} [See Figure \ref{fig:simpaths}]
Two critical paths $\pi_{v, c}, \pi_{v', c'} \in \Pi$ are \emph{similar} if (1) $c=c'$ and (2) there is a vector $b$ orthogonal to $c$ with $4 \|b\|_2 \le w$ and $v + b = v'$ (where this vector addition is performed modularly in $\zz_h \times \zz_w$).
\end{definition}

\begin{figure} [ht]
\begin{center}
\begin{tikzpicture}
\draw [ultra thick] (0, 0) rectangle (12, 2);
\draw [fill=black] (6, 1) circle [radius=0.2];
\node at (-0.5, 1) {$h$};
\node at (6, -0.5) {$w$};
\draw [thick, dashed] (5, 2) -- (7, 0);
\draw [thick, dashed] (3, 2) -- (5, 0);
\draw [thick, dashed] (7, 2) -- (9, 0);
\draw [ultra thick, ->] (6, 1) -- (6.5, 1.5);
\node at (6.5, 1.1) {$c$};
\node at (6, 0.6) {$v$};
\end{tikzpicture}
\end{center}
\caption{\label{fig:simpaths} The critical paths similar to $\pi_{v, c}$ are precisely those that start at any point $v' \in \zz^2$ on one of the dashed lines and which use the same critical vector $c$.}
\end{figure}

In order to understand this definition intuitively, let us imagine that $c = (0, 5)$ (vertical) and $b = (1, 0)$ (horizontal).\footnote{Technically this isn't a valid choice of $c$ since we restricted our critical vectors $c$ to the cone between $(1, 2)$ and $(2, 1)$, but we will briefly forget this technical detail in order to communicate intuition.}
Consider the critical paths $\pi_{(0, 0), c}$ and $\pi_{(1, 0), c}$, which end at the points $[(0, 5\ell), \ell)]$ and $[(1, 5\ell), \ell]$ respectively.
Uniqueness of these critical paths follows directly from the convexity of our critical vectors: every critical vector besides $c$ will have second coordinate $<5$, and thus the only way to gain $5\ell$ points in the second coordinate using $\ell$ steps in our graph is to walk the edge corresponding to $c$ at each step.
But, we notice that this exact same argument implies that there is no path with endpoints $[(0, 0), 0] \leadsto [(1, 5\ell), \ell]$.
Such a path would have to gain $5\ell$ points in its second coordinate using $\ell$ steps; the only way to do so is to walk the edge corresponding to $c$ at each step, but this uniquely defines a path that ends at the point $[(0, 5\ell), \ell]$ rather than $[(1, 5\ell), \ell]$.
By identical logic, there can also be no path with endpoints $[(1, 0), 0] \leadsto [(0, 5\ell), \ell]$.
So it is indeed safe to assign these two critical paths the same source node $s \in S$.

This same basic argument applies for any two critical paths that:
\begin{itemize}
\item Use the same critical vector $c$,
\item Have starting points that differ by a vector orthogonal to $c$ (mod $h, w$ in their respective coordinates), and
\item Obey a certain ``non-rollover'' property: for example, the previous argument might fail if $5\ell < h$, and thus we could entertain $[(0, 0), 0] \leadsto [(1, 5\ell), \ell]$ paths that only gain (say) $5\ell - h$ points in the second coordinate, which still reach the same endpoint since the second coordinate operates mod $h$.
\end{itemize}

This third property is the most technical one to analyze.
It is executed below by ``unrolling'' the construction, moving from the modular space $\zz_h \times \zz_w$ into $\rr^2$ for the sake of analysis.
When we try to prove uniqueness for a class of similar paths $\Sigma$, we will first take the set of points in $\zz_h \times \zz_w$ used as start nodes of paths in $\Sigma$, and we will embed these into a line segment $L \subseteq \rr^2$.
Then we define the ``unrolled space'' $U \subseteq \rr^2$ to be the points that can be reached by starting from any point on $L$, and adding a vector in the cone between $(1, 2), (2, 1)$ of length $\le \ell r$.
In particular, $U$ contains all points that can be reached by starting at some point $\ell \in L$ and adding on $\le \ell$ critical vectors.
The ``non-rollover'' property is that it will turn out that no two points in $U$ are equivalent mod $h, w$, and thus we do not have to worry about modular rollover in our analysis.
This property is enforced by the various parameter inequalities sprinkled through our argument; for example that $3 \ell r < h$, and it is exactly the reason why we restrict our critical vectors to the cone between $(1, 2)$ and $(2, 1)$.

\paragraph{The Rest of the Argument Formalized.}

We now execute the steps outlined above, with full formal detail.

\begin{lemma}
We can build a partition $S$ of $\Pi$ into
$$ | S | = O\left(\frac{|\Pi| r}{w} \right) $$
parts such that any two critical paths in the same part are similar.
We will call these parts \emph{similar classes}.
\end{lemma}
\begin{proof}
For each critical path $\pi_{v, c}$, we define an associated set of critical paths $\Sigma_{v, c}$ as follows.
Let $b$ be a vector orthogonal to $c$ with $\|b\|_2 = \|c\|_2$; that is, if $c = (c_1, c_2)$, then we can take $b := (-c_2, c_1)$.
We then define
$$\Sigma_{v, c} := \left\{\pi_{v + ib, c} \in \Pi \ \mid \ 0 \le i < \frac{w}{8 r} \right\}.$$
Notice that all the paths in $\Sigma_{v, c}$ are similar to $\pi_{v, c}$, but actually $\Sigma_{v, c}$ only contains about half of \emph{all} the paths similar to $\pi_{v, c}$: these parameter restrictions enforce $8 \|b\|_2 \le w$, rather than $4 \|b\|_2 \le w$ as in the original definition.

Now we build our partition $S$ of $\Pi$ iteratively, in two stages.
Initially $S = \emptyset$.
While there is a set $\Sigma_{v, c}$ that is disjoint from every set currently in $S$, add $\Sigma_{v, c}$ to $S$.
Once this terminates, we then iterate through each critical path $\pi_{v,c} \in \Pi$ that is not yet contained in any set in $S$.
For each such path, since we did \emph{not} add $\Sigma_{v, c}$ to $S$, we must have previously added a set that intersects $\Sigma_{v, c}$ to $S$.
But since $\pi_{v, c}$ itself is not in that set, we must have previously added a set of the form $\Sigma_{v +ib, c}$ to $S$ (where as above, $b$ is orthogonal to $c$ with $\|b\|_2 = \|c\|_2$, and $0 \le i < \frac{w}{8r}$).
We then choose to insert $\pi_{v, c}$ into $\Sigma_{v + ib, c} \in S$.

To prove correctness: by construction $S$ is a partition of $\Pi$, and each part has $\Omega(w/ r)$ paths, leading to the claimed bound on the size $|S|$.
It is immediate from the definition that, if two paths are contained in the same set $\Sigma_{v, c}$ in the first stage of the construction, then they are similar.
Moreover, whenever we add a path $\pi_{v, c}$ to a set $\Sigma_{v + ib, c}$ in the second stage of the construction, for any other path $\pi_{v + ib + i'b, c} \in \Sigma_{v+ib, c}$, we have
$$i + i' < \frac{w}{8 r} + \frac{w}{8 r} = \frac{w}{4 r},$$
and thus $\pi_{v, c}$ is indeed similar to each path in $\Sigma_{v + ib, c}$.
\end{proof}

We are now going to return to our construction and add a set of additional nodes $S$, where each node $\Sigma \in S$ represents one similar class from the previous lemma.
For each similar class $\Sigma \in S$ and each critical path $\pi_{v, c} \in \Sigma$, the \emph{extended critical path} $\pi^*_{v, c}$ is the one that uses $\Sigma$ as its first node and then the critical path $\pi_{v, c}$ after that.
We also add edges of the form $(\Sigma, v)$ so that the extended critical paths are indeed paths in our graph.
This completes the construction.
Its key property is:

\begin{lemma} \label{lem:crituniq}
Every extended critical path is the unique path between its endpoints.
\end{lemma}
\begin{proof}
Let $\pi^*_{v, c}$ be an extended critical path and let $\Sigma$ be its similar class.
We will say that a point $u \in \zz_h \times \zz_w$ is \emph{relevant} if there exists a path in our graph from $\Sigma$ to a point of the form $[u, i]$, or it is \emph{irrelevant} otherwise.
Although we have so far described our graph construction in the modular space $\zz_h \times \zz_w$, it will now be helpful to ``unroll'' the construction and associate each relevant point $u \in \zz_h \times \zz_w$ to a single point in $\zz^2$ which is equivalent to $u$ (mod $h, w$ in its respective coordinates).
Specifically, we will treat $v$ itself as the point in $\rr^2$ that happens to have coordinates in $\{0, \dots, h-1\} \times \{0, \dots, w\}$.
Then we define a line segment
$$L := \left\{v + b \in \rr^2 \ \mid \ b \perp c, \|b\|_2 \le \frac{w}{4} \right\}.$$
Finally, let $U \subseteq \rr^2$ be the set of points that can be written as $y + z$ where $y \in L$ and $z$ is a vector of length $\le \ell r$ that lies in the acute cone between $(1, 2)$ and $(2, 1)$.
The important property of $U$ is that for every relevant $u \in \zz_h \times \zz_w$, there is exactly one corresponding point in $U$ that is equivalent to $u$ (mod $h, w$).
To see this, we observe:
\begin{itemize}
\item By construction, every relevant point $u \in \zz_h \times \zz_w$ can be written as the sum of (1) a point $y$ that is equivalent to one of the points in $L$ (mod $h, w$), and (2) a sum $z$ of $\le \ell$ critical vectors.
Since the critical vectors lie in the cone between $(1, 2)$ and $(2, 1)$ and they each have length $\le r$, it follows that $z$ is a vector in the cone between $(1, 2)$ and $(2, 1)$ with length $\le \ell r$.
Thus, $u$ is equivalent to a point in $U$.

\item We now argue that no two distinct points $u_1, u_2 \in U$ are equivalent to each other (mod $h, w$), and so the above equivalence is unique.
Suppose $u_1 \equiv u_2$, and let $u_1 = y_1 + z_1, u_2 = y_2 + z_2$, where as before $y_1, y_2 \in L$ and $z_1, z_2$ are both vectors in the cone between $(1, 2), (2, 1)$ of length $\|z_1\|_2, \|z_2\|_2 \le \ell r$.
By the triangle inequality all points $u \in U$ satisfy
$$\dist(v, u) \le \frac{w}{4} + \ell r < \frac{w}{4} + \frac{w}{2} = \frac{3w}{4},$$
and thus if $u_1 \equiv u_2$ they must in fact have the same second (width) coordinate in $\rr^2$, since their second coordinates cannot differ by a multiple of $w$ and still satisfy this distance inequality.
Noting that $\|(1, 2)\|_2 = \|(2, 1)\|_2 = \sqrt{5}$, the width coordinates of $y_1, y_2$ thus differ by at most
$$\frac{2\ell r}{\sqrt{5}} < \frac{h}{2 \sqrt{5}}.$$
Since $y_1, y_2$ both lie on the line $L$, which is orthogonal to $c$ between $(1, 2)$ and $(2, 1)$, the height coordinates of $y_1, y_2$ thus differ by at most
$$\frac{h}{2 \sqrt{5}} + \frac{2\ell r}{ \sqrt{5}} < \frac{h}{\sqrt{5}}.$$
So the height coordinates of $u_1, u_2$ differ by at most
$$\frac{h}{\sqrt{5}} + \frac{2\ell r}{\sqrt{5}} < \frac{3h}{2 \sqrt{5}} < h.$$
So the height coordinates of $u_1, u_2$ cannot differ by a multiple of $h$, so they are equivalent (mod $h$) iff they are equal.
We have now proved that if $u_1 \equiv u_2$ then $u_1 = u_2$, giving uniqueness as desired.
\end{itemize}


The remainder of the proof treats points as vectors in $U \subseteq \rr^2$, and it follows a geometric potential argument along the lines of \cite{Behrend46, CE06}.
Let $L'$ be the line in $\rr^2$ that contains $L$ as a subsegment, and define a potential function $\phi(u)$ over points $u \in U$ to be the Euclidean distance from $u$ to the closest point on $L$.
Now let us consider any path
$$q := \left(\Sigma, u_0, u_1, \dots, u_{\ell} = v + \ell c \right)$$
(which has the same endpoints as the extended critical path $\pi^*_{v, c}$ for which we are trying to prove uniqueness).
Notice that, by definition of similarity, we must have $u_0 \in L$, and so $\phi(u_0) = 0$.
At the other endpoint, since $c \perp L$, we have $\phi(v + \ell c) = \ell \|c\|_2$.
So between the second and last node of $q$, the potential must increase by $\ell \|c\|_2$ in total.
We can then exploit convexity: for any edge along $q$ of the form $([u_i, i], [u_{i+1}, i+1])$, we have
$$\phi(u_{i+1}) - \phi(u_i) \le \|c\|_2,$$
with equality iff $u_{i+1} = u_i + c$.
Hence, the only way to gain $\ell \|c\|_2$ potential using the $\ell$ edges along $q$ between $[u_0, 0]$ and $[u_{\ell} = v + \ell c, \ell]$ is if we have $u_{i+1} = u_i + c$ for all $i$.
This also implies that $u_0 + \ell c = v + \ell c$, and so $u_0 = v$.
We have thus proved that $q = \pi^*_{v, c}$, and thus $\pi^*_{v, c}$ is the unique path between its endpoints.
\end{proof}

\begin{theorem} \label{thm:main-lb}
In the parameter range $n^4 \sigma^6 \le p^9$, we have $\rpp(n, p, \sigma) = \Omega\left( n^{4/5} p^{1/5} \sigma^{1/5} \right)$.
\end{theorem}
\begin{proof}
We take our demand pairs $P$ to be the endpoints of all extended critical paths.
Hence $P \subseteq S \times V$.
Additionally, we will set parameters such that $\ell := \left\lceil h/r \right\rceil - 1$ (so the $\ell r < h$ parameter restriction holds).
The remaining parameters are set as follows:
\begin{align*}
h &:= n^{1/2} \sigma^{1/2} p^{-1/2}\\
w &:= n^{-3/10} \sigma^{-7/10} p^{13/10}\\
r &:= n^{-3/10} \sigma^{3/10} p^{3/10}.
\end{align*}
We then recap the parameters of our construction.
The number of demand pairs is
$$|\Pi| = \Theta \left(hwr^{2/3}\right) = \Theta(p).$$
The number of source nodes is
$$|S| = O\left(\frac{|\Pi| r}{w}\right) = O\left(\frac{ hwr^{5/3}}{w}\right) = O\left( hr^{5/3} \right) = \Theta(\sigma).$$
The number of nodes is
$$n = |S| + (\ell+1)hw = O\left(\ell hr^{2/3}\right) + O(\ell hw) = O(\ell hw) = O\left( h^2 w r^{-1} \right) = \Theta(n).$$
The number of edges in our construction is
$$|E| = \Omega\left( \ell |\Pi| \right) = \Omega\left( \ell hwr^{2/3} \right) = \Omega\left( h^2 w r^{-1/3}\right) = \Omega\left( n^{4/5} p^{1/5} \sigma^{1/5}\right).$$

The edges of the graph are exactly the edges contained in any extended critical path.
So by Lemma \ref{lem:crituniq}, all edges in the graph must remain in the associated reachability preserver.
Finally, we need to address the restriction in range of parameters inherited from the inequalities $1 \le r$ and $4r \le h \le w$.
The former inequality
$$r = n^{-3/10} \sigma^{3/10} p^{3/10} \ge 1$$
is harmless: it requires that $\sigma p \ge n$, but in the range $\sigma p < n$ our lower bound is no better than $\Omega(n)$, which already holds trivially.
The inequality $4r \le h$ gives
\begin{align*}
4n^{-3/10} \sigma^{3/10} p^{3/10} &\le n^{1/2} \sigma^{1/2} p^{-1/2}\\
4p^{4/5} &\le n^{4/5} \sigma^{4/5}\\
4p &\le n \sigma
\end{align*}
which is also harmless (since $P \subseteq S \times V$ we already have that the number of demand pairs is $O(n \sigma)$, so this inequality only affects implicit constant factors).
Finally, the inequality $h \le w$ gives
\begin{align*}
n^{1/2} \sigma^{1/2} p^{-1/2} &\le n^{-3/10} \sigma^{-7/10} p^{13/10}\\
n^{4/5} \sigma^{6/5} &\le p^{9/5}\\
n^4 \sigma^6 &\le p^9,
\end{align*}
completing the proof.
\end{proof}

\medskip
\paragraph{Acknowledgement.}
We thank several anonymous reviewers for exceptionally helpful comments on the writing of the paper.
We thank Seth Pettie for suggesting a simplification to our lower bound construction, and we thank Nicole Wein, Virginia Vassilevska Williams, and Zixuan Xu for finding a mistake in an earlier draft.

The authors were supported by the grants of Virginia Vassilevska Williams: NSF Grants CCF-1417238, CCF-1528078 and CCF-1514339, and BSF Grant BSF:2012338.
Part of the work was performed while visiting the Simons Institute for the Theory of Computing, Berkeley, CA.

\bibliography{../../../../BIB/references}
	\bibliographystyle{plain}

\end{document}